\documentclass[11pt,a4paper]{amsart}

\usepackage{times,epsfig}
\usepackage{amsfonts,amscd,amsmath,amssymb,amsthm}

\usepackage{enumitem}
\usepackage{bbm}

\usepackage{tikz}
\usepackage{xy} \xyoption{all}

\usepackage{color,graphics}

\usepackage{bbm}

\usepackage{hyperref}

\numberwithin{equation}{section}

\newcommand{\uloopr}[1]{\ar@'{@+{[0,0]+(-4,5)}@+{[0,0]+(0,10)}@+{[0,0] +(4,5)}}^{#1}}
\newcommand{\uloopd}[1]{\ar@'{@+{[0,0]+(5,4)}@+{[0,0]+(10,0)}@+{[0,0]+ (5,-4)}}^{#1}}
\newcommand{\dloopr}[1]{\ar@'{@+{[0,0]+(-4,-5)}@+{[0,0]+(0,-10)}@+{[0, 0]+(4,-5)}}_{#1}}
\newcommand{\dloopd}[1]{\ar@'{@+{[0,0]+(-5,4)}@+{[0,0]+(-10,0)}@+{[0,0 ]+(-5,-4)}}_{#1}}

\newcommand{\luloop}[1]{\ar@'{@+{[0,0]+(-8,2)}@+{[0,0]+(-10,10)}@+{[0, 0]+(2,2)}}^{#1}}

\DeclareSymbolFont{SY}{U}{psy}{m}{n}
\DeclareMathSymbol{\emptyset}{\mathord}{SY}{'306}

\DeclareMathSymbol{\newtimes}{\mathbin}{SY}{'264}

\newcommand{\fh}{\mathfrak{h}}

\newcommand{\fA}{\mathfrak{A}}

\newcommand{\cA}{{\mathcal A}}
\newcommand{\cB}{{\mathcal B}}

\newcommand{\cF}{{\mathcal F}}
\newcommand{\cG}{{\mathcal G}}         
\newcommand{\cH}{{\mathcal H}}
\newcommand{\cI}{{\mathcal I}}
\newcommand{\cJ}{{\mathcal J}}
\newcommand{\cK}{{\mathcal K}}

\newcommand{\cN}{{\mathcal N}}
\newcommand{\cO}{{\mathcal O}}
\newcommand{\cP}{{\mathcal P}}

\newcommand{\cR}{{\mathcal R}}

\newcommand{\cT}{{\mathcal T}}
\newcommand{\cU}{{\mathcal U}}

\newcommand{\cZ}{{\mathcal Z}}

\renewcommand{\1}{\mathbbm 1}

{\bf}{\it}
{\bf}{\it}
{\bf}{\it}
{\bf}{\it}

{\bf}{\it}

\newtheorem{theorem}{Theorem}[section]{\bf}{\it}
\newtheorem{proposition}[theorem]{Proposition}{\bf}{\it}
{\bf}{\it}
\newtheorem{lemma}[theorem]{Lemma}{\bf}{\it}
{\bf}{\it}
{\bf}{\it}
{\bf}{\it}

\theoremstyle{definition}
\newtheorem{definition}[theorem]{Definition}
\newtheorem{remark}[theorem]{Remark}
\newtheorem{example}[theorem]{Example}

\theoremstyle{plain}

\theoremstyle{definition}

\newlist{Fenumerate}{enumerate}{1}
\setlist[Fenumerate, 1]{label = \textbf{(F\arabic*)}, ref = {{\bf (F\arabic*)}}}

\newcommand{\R}{\mathbb{R}}

\newcommand{\N}{\mathbb{N}}
\newcommand{\C}{\mathbb{C}}

\newcommand{\BH}{{\mathcal{B}(\mathcal{H})}}

\begin{document}

\title{Notions of infinity in quantum physics}

\author[Fernando Lled\'{o}]{Fernando Lled\'{o}$^{1}$}
\address{Department of Mathematics, University Carlos~III Madrid,
  Avda.~de la Universidad~30, 28911 Legan\'{e}s (Madrid), Spain
  and Instituto de Ciencias Matem\'{a}ticas (CSIC - UAM - UC3M - UCM).}
\email{flledo@math.uc3m.es}

\author[Diego Mart\'{i}nez]{Diego Mart\'{i}nez$^{2}$}
\address{Department of Mathematics, University Carlos~III Madrid,
  Avda.~de la Universidad~30, 28911 Legan\'{e}s (Madrid), Spain
  and Instituto de Ciencias Matem\'{a}ticas (CSIC - UAM - UC3M - UCM).}
\email{lumartin@math.uc3m.es}

\thanks{Supported by research projects MTM2017-84098-P and Severo Ochoa SEV-2015-0554 of the Spanish Ministry of Economy and Competition
(MINECO), Spain.}

\date{\today}


\maketitle

\begin{center}
 {\em Dedicated to Alberto Ibort on the occasion of his 60th birthday}
\end{center}

\begin{abstract}
In this article we will review some notions of infiniteness that appear in Hilbert space operators and operator algebras.
These include proper infiniteness, Murray von Neumann's classification into
type~$I$ and type~$III$ factors and the class of F\o lner C*-algebras that capture some aspects of amenability. 
We will also mention how these notions reappear 
in the description of certain mathematical aspects of quantum mechanics, quantum field theory and the theory of 
superselection sectors. We also show that the algebra of the canonical anti-commutation relations
(CAR-algebra) is in the class of F\o lner C*-algebras.
\end{abstract}

\section{Introduction}
\label{sec:1}

In this article we will review some situations in which different notions of infinity
manifest in quantum mechanics and quantum field theory. 
To begin let us recall some reasonable and basic definitions of finiteness in set theory 
(cf., \cite[Introduction]{KL15}). A set $X$ can be called {\em finite} if any of these conditions holds:
\begin{Fenumerate}
 \item \label{F1} there is a bijection $\varphi\colon X\to\{1,\dots, n\}$ for some $n\in\N$;
 \item \label{F2} there does not exist a (disjoint) partition $X=X_1\sqcup X_2$ such that $|X|=|X_1|=|X_2|$, where $|\cdot|$ denotes the cardinality
of a set;
 \item \label{F3} every injective map $f\colon X\to X$ is surjective.
\end{Fenumerate}

The characterization \ref{F1} uses the external structure of the natural numbers and is constructive, while \ref{F2} identifies 
finiteness through the absence of a certain kind of decomposition, which resembles a \textit{paradoxical} decomposition. 
The last item \ref{F3} refers to Dedekind's definition of finiteness and is intrinsic to the structure.
All these ideas and, in particular, their negation, reappear in a very natural way in the context of linear operators and operator algebras. This is
how they also enter in the description of some aspects of Quantum Theory.

For infinite sets on which, in addition, a dynamic is defined one can further classify the system according to the 
dichotomy {\em amenable} versus {\em paradoxical}. It must be highlighted that \textit{dynamics} is here understood in a wide sense, such as the
action of a group on a space or as the action of an algebra on itself by left multiplication. 
The idea of amenability was introduced in the context of group actions by von Neumann in 1929 (cf. \cite{Neumann29}) and its absence in the action of
the rotation group on the unit ball $B_1\subset\R^3$
was recognized as a fundamental reason that explains the possibility of paradoxically decomposing $B_1$. This fact eventually came to be known as the
\textit{Banach-Tarski paradox} (cf., \cite{Wagon16,bPaterson88,bRunde02}). 
Since then this dichotomy amenable versus paradoxical has enriched many other fields including algebras, metric spaces
and operator algebras. Roughly speaking, amenable structures have an 
internal approximation in terms of finite substructures (the so-called
\textit{F\o lner sequences}) that have controlled growth with respect to the dynamics considered. It is therefore clear that all finite structures
are normally amenable, while infinite structures might be or not. We refer to 
\cite{KL15,Cec01,Silberstein-Grigorchuk-Harpe-99,Day57,Elek03,Gromov99,ALLW-1,ALLW-2} 
for additional motivation and results on this body of work. 

The aim of this article is to review some results showing the different \textit{degrees of infiniteness} that appears in some situations in Quantum
Theory.
We also bring into this analysis the class of F\o lner C*-algebras that capture some aspects of amenability in the context of operator 
algebras. These algebras can be characterized in terms of a sequence of unital completely positive linear 
maps into matrices which are asymptotically multiplicative. We will show that the CAR-algebra is in this class.
We begin reviewing some notions of infiniteness that appear in the description of Hilbert space operators and operator algebras. 
In particular we introduce notions of proper infiniteness and Murray von Neumann's classification into
type~$I$ and type~$III$ factors. We also recall some important results in local quantum physics in relation to this topic, in particular, Borchers
property or the construction of the field algebra in the theory of superselection sectors.

\section{Operators and operator algebras in Hilbert spaces}
\label{sec:2}

Let $\cH$ be a complex separable Hilbert space and denote by $\BH$ the set of all bounded linear operators on $\cH$. Given an operator $T \in \BH$,
its \textit{operator norm} is given by
 \begin{equation}\label{eq:op-norm}
  \|T\| :=\sup_{\|x\|=1} \|Tx\|\;,
 \end{equation}
where $\|x\|$ is the Hilbert space norm of the vector $x \in\cH$ induced by the scalar product $\langle\cdot,\cdot\rangle$.

\begin{example}\label{ex:1}
 \begin{itemize}
  \item[(i)] If $\cH\cong\C^n$, then $\BH\cong M_n(\C)$. In this case, it is well known that any isometry is necessarily a unitary, i.e., 
  for any $M\in\BH$ with $M^*M=\1$, then $MM^*=\1$. This realizes Dedekind's notion \ref{F3} of finiteness in the context of linear maps, since any
injective map must as well be surjective.
  \item[(ii)] If $\cH\cong\ell_2(\N)$ (the Hilbert space of square summable sequences), denote its canonical basis by $\{e_i\}_{i\in\N}$. 
   The infinite dimension of the Hilbert space has now several consequences that can be understood as a linear analogy to
   Hilbert's Hotel. The following examples of non-unitary isometries can be understood as a negation of the finiteness condition $\ref{F3}$
   in the linear context.
   \begin{itemize}
    \item[a)] {\em Unilateral shift:} Let $S e_i:=e_{i+1}$, $i\in\N$, i.e., 
                               $S\cong
                               \begin{pmatrix}
                                  0 & 0 & 0 & \dots \\
                                  1 & 0 & 0 & \dots \\
                                  0 & 1 & 0 & \dots \\
                                  \vdots  &  \ddots  & \ddots  & \ddots
                               \end{pmatrix}$. 
              Then we have $S^*S=\1$, but $SS^*=\1-P_0$, where $P_0\left(\cdot\right):=\langle e_0,\cdot\rangle \,e_0$ is the range one projection
onto the 
              linear subspace $\C \cdot e_0$. In this case one says that $\1$ is an infinite projection (see Definition~\ref{def:equiv} below).

    \item[b)]\label{ex:1:ii:b} {\em Generators of the Cuntz algebra:} define $S_1e_i:=e_{2i}$ and $S_2e_i:=e_{2i+1}$. These are isometries 
              (i.e., $S_1^*S_1=S_2^*S_2=\1$) and satisfy, in addition, 
              \[
                      S_1^*S_2=0\quad\mathrm{and}\quad S_1S_1^*+S_2S_2^*=\1\;.
              \]
              In other words, the ranges of $S_1$ and $S_2$ are infinite dimensional and mutually orthogonal subspaces of $\ell^2(\N)$, giving a
negation 
              of the finiteness condition \ref{F2}. In this case one says that $\1$ is a properly infinite projection
              (see Definition~\ref{def:equiv} below).
   \end{itemize}

  \item[(iii)] {\em Partial isometries:} A linear map $V\colon\cH\to\cH$ is a partial 
  isometry if $V^*V$ is an orthogonal projection, which is called {\em domain projection}. 
  This condition directly implies that $VV^*$ is also a projection, the
  so-called {\em range projection}. These partial isometries are a generalization of the notion of isometry.
 \end{itemize}

\end{example}

Next we introduce two types of operator algebras that will be important for this article, namely, C*- and von Neumann algebras. 
General references on this topic are, e.g., \cite{bDavidson96} or \cite[Chapter~2]{bBratteli02}.
We call a *-subalgebra $\cA\subset\BH$ a {\em C*-algebra} if it is 
closed with respect to the uniform topology, i.e., the
topology defined by the operator norm $\|\cdot\|$ (cf., Eq.~(\ref{eq:op-norm})). 
Important examples of C*-algebras are those generated by isometries having
mutually orthogonal ranges. For
$n \geq 2$, the Cuntz algebra $\mathcal{O}_n$ is the essentially unique C*-algebra generated by isometries
$S_1,\dots,S_n$ satisfying
\[
 S_i^*S_j=\delta_{ij} \1,\; i,j\in\N\;,\quad\mathrm{and}\quad \sum_{i=1}^n S_iS_i^*=\1\;.
\]
Example \ref{ex:1:ii:b} shows how these isometries can be realized as elements of $\cB(\ell_2(\N))$.

A unital *-subalgebra $\cN\subset\BH$ is a {\em von Neumann algebra} if it is closed under the weak operator topology.
A useful and alternative way to understand this class of algebras is through the notion of commutant of a set of operators.
If $\mathcal{S}$ is a self-adjoint subset of $\BH$ 
(i.e., if $S\in\mathcal{S}\subset\BH$, then $S^*\in\mathcal{S}$), then 
we denote by $\mathcal{S}'$ the {\em commutant} of $\mathcal{S}$
in $\BH$, i.e., the set of all operators in 
$\BH$ commuting with all elements in $\mathcal{S}$. 
Von Neumann's celebrated bicommutant theorem shows that a unital *-subalgebra $\cN\subset\BH$ is a von Neumann algebra
iff $\cN=\cN''$. Therefore, if $\mathcal{S}$ is a self-adjoint subset of $\BH$, then
$\mathcal{S}''$ is the smallest von Neumann algebra
containing $\mathcal{S}$. A von Neumann algebra $\cN$ is called a {\em factor} if it has a trivial center, i.e., if $\cN\cap\cN'=\mathbb{C} \cdot \1$.

Any von Neumann algebra is generated as a norm-closed space by the set of its projections, which we denote by $\cP(\cN)$. Therefore,
the classification we are interested in of von Neumann algebras is based on the classification of $\cP(\cN)$.
For the purpose of this article, it is enough to assume that the von Neumann algebra $\cN$ is a (nonzero) factor, since general von Neumann algebras
can be canonically decomposed into factors.

\begin{definition}\label{def:equiv}
Let $\cN$ be a factor and denote by $\cP(\cN)$ its lattice of orthogonal projections in $\cN$. All the following definitions are \textit{modulo
$\cN$}, that is, depend on $\cN$. For $P, Q  \in \cP(\cN)$ we say
\begin{enumerate}
\item[(i)] 
$P$ is \emph{minimal} if $P \neq 0$ and for any projection $P_0 \in \cP(\cN)$, $P_0 \leq P$ implies either $P_0 = 0$ or $P_0 = P$.
\item[(ii)] $P \sim Q$ if there exists a partial isometry $V \in \cN$ such that 
$P = V^*V$ and $Q = VV^*$. 
The relation $\sim$ is called also Murray von Neumann equivalence.

\item[(iii)] $P$ is \emph{finite} ($\mathrm{mod}\,\cN$) if the only projection $P_0 \in \cP(\cN)$ with $P \sim P_0 \leq P$ is the projection $P$
itself.
  
  If $P$ is not finite then it is called \emph{infinite} ($\mathrm{mod}\,\cN$). That is, there is a $P_0 \in \cP(\cN)$ such that $P \sim P_0 < P$,
namely, $P$ is equivalent to a proper subprojection of itself.

  $P$ is \textit{properly infinite} if there exist $P_1,P_2\in\cP(\cN)$
  such that $P\sim P_1 \sim P_2$, $P_1 + P_2 \leq P$ and $P_1 P_2 = 0$, i.e., $P_1\cH\perp P_2\cH$.

\item[(iv)] A factor $\cN$ is called \emph{finite} (respectively, \emph{infinite} or \emph{properly infinite}) if $\1$ is a finite (respectively, 
infinite or properly infinite) projection.
\end{enumerate}
\end{definition}

\begin{remark}\label{rem:min}
\begin{enumerate}
 \item[a)] The definition of finite, infinite and properly infinite projections can be stated similarly in the 
 context of C*-algebras. 
It is clear from Example~\ref{ex:1} that $\1\in M_n(\C)$ is a finite projection. On the contrary
$\cB(\ell_2(\N))$ is an infinite C*-algebra via the equivalence $\1\sim \1-P_0 < \1$.

Finally, the Cuntz algebras $\cO_n$ (and any C$^*$-algebra containing them), are \textit{the} prototypes of properly infinite C$^*$-algebras, since we
have from Example~\ref{ex:1} that $S_1 S_1^* + \dots + S_n S_n^* = \1$ while
\[
 \1 = S_1^*S_1 = \dots = S_n^*S_n \quad \mathrm{and} \quad S_i^*S_j = \delta_{ij} \1.
\]

\item[b)]
It follows from the definition that any minimal projection in a von Neumann algebra is automatically
finite. The most prominent example of minimal projection is the range one projection $P_x(\cdot):=\langle x,\cdot\rangle \,x$, defined for any
$x\in\cH$ with $\|x\|=1$. 

It should be noted that if $P$ is a minimal projection in a von Neumann algebra
$\cN$, then the corner algebra is one-dimensional, i.e., $P\cN P=\C P$.
Moreover, all minimal projections are equivalent.
\end{enumerate}
\end{remark}

According to the properties of the lattice of projections we mention next some large subclasses of factors.
\begin{definition}
 Let $\cN$ be a factor and $\cP(\cN)$ its lattice of projections.
 \begin{itemize}
  \item[(i)] $\cN$ is said to be of {\em type~I} if $\cP(\cN)$ contains a minimal nonzero projection.
  \item[(ii)] $\cN$ is said to be of {\em type~III} if $\cP(\cN)$ contains no nonzero finite projection. 
 \end{itemize}
\end{definition}

Type~$III$ factors show, roughly speaking, the highest degree of infiniteness. In fact, for this class of algebras
any nonzero projection admits the following halving property (which can be understood as a negation of \ref{F2}
in the linear context).

\begin{lemma}
 Let $\cN$ be a factor and $\cP(\cN)$ its lattice of projections. Then $P\in\cP(\cN)$ is infinite if and only if
 $P$ admits the following decomposition 
 \[ 
  P=(P-Q)+Q\quad\mathrm{for~some}\quad Q\leq P\quad\mathrm{and}\quad P\sim Q\sim (P-Q)\;.
 \]
\end{lemma}

For simplicity we will focus in this article only on these two classes of factors. Type~$II$ factors (those having no minimal projections
but having nonzero finite projections) are also important in describing certain aspects quantum theory 
(see, e.g., \cite{Keyl06,Naaijkens11}).

\section{F\o lner C*-algebras}
\label{sec:3}

Motivated by the dichotomy amenable versus paradoxical in group theory we will introduce in this section the class of 
F\o lner C*-algebras. These algebras correspond to the amenable groups, in the sense of having a good 
internal approximation in terms of matrices that have controlled growth with respect to the dynamics given by the product.
We will also define the notion of algebraic amenability and some relation to the
class of F\o lner C*-algebras.
These ideas will be used in the next section.

For the next definition, recall that a tracial state on a C$^*$-algebra $\cA$ is a positive and normalized functional $\tau\colon\cA\to\C$
that satisfies the usual tracial property $\tau(AB)=\tau(BA)$ for any $A,B\in\cA$. In the next definition we specify 
the subclass of amenable traces (see, e.g., \cite[Chapter~6]{bBrown08}).

\begin{definition}
Let $\cA\subset\BH$ be a unital and separable C*-algebra. $\cA$ is called a {\em F\o lner C*-algebra} if it has an {\em amenable 
trace} $\tau$, i.e., a tracial state on $\cA$ that extends to a state $\psi$ on $\BH$ that has $\cA$ in its centralizer, i.e.,
\[
   \tau=\psi|_{\cA}\quad\mathrm{and}\quad \psi(XA)=\psi(AX)\;, \;A\in\cA\;,\; X\in\BH\;.
\]
\end{definition}

From this definition it follows immediately that any unital C$^*$-subalgebra of
a F\o lner C$^*$-algebra is again in this class
and that any finite dimensional algebra is a F\o lner C$^*$-algebra, since the usual normalized trace of a matrix will do.

\begin{remark}
 The state $\psi$ in the preceding definition is called {\em hypertrace} in the literature and this class of algebras is also referred
 as weakly hypertracial (see \cite{Bedos95} and references therein). The preceding definition is equivalent to the intrinsic definition
 of an abstract F\o lner C$^*$-algebra $\cA$ in terms of a sequence of unital completely positive linear 
 maps into matrices $\varphi_n\colon\cA\to M_{k(n)}$ which are asymptotically multiplicative. This approach shows explicitly
 the finite approximation scheme of this class of algebras (cf., \cite[Theorem~4.3]{AL14}). Moreover, this class of algebras
 are also relevant in problems of spectral approximation (cf., \cite{Arveson94,Bedos97,Lledo13,LledoYakubovich13}).
\end{remark}

We will conclude by introducing the notion of algebraically amenable algebras. We will restrict to the case of subalgebras of 
C$^*$-algebras, but the definition and results are true for arbitrary algebras over arbitrary fields (cf. \cite{Elek03,ALLW-1,Gromov99}).

\begin{definition}\label{def:alg-amenable}
Let $\fA\subset\cA$ be a *-subalgebra of a C$^*$-algebra $\cA$. We call $\mathfrak{A}$ algebraically amenable if there is 
a sequence $\left\{W_k\right\}_{k = 1}^{\infty}$ of finite dimensional subspaces of $\mathfrak{A}$ satisfying 
\[
 \lim_{k\to\infty}\frac{\mathrm{dim}(AW_k+W_k)}{\mathrm{dim}(W_k)}=1\;, \quad A\in \mathfrak{A}\;.
\]
\end{definition}

Next we mention an important relation between algebraic amenability and the
class of F\o lner C$^*$-algebras.
For a complete proof we refer to \cite[Theorem~3.17]{ALLW-2}.
\begin{theorem}\label{teo:alg-amenable-Foelner}
Let $\mathfrak{A}\subset \cA$ be a dense *-subalgebra of a unital separable C$^*$-algebra $\cA$.
If $\mathfrak{A}$ is algebraically amenable, then $\cA$ is a F\o lner C$^*$-algebra.
\end{theorem}

\section{Quantum physics}
\label{sec:4}

In the mathematical description of a physical theory one needs to specify the set of observables, the set of
states and, possibly, the family of symmetries of the theory, typically described in terms of a group action. 
For a description of a quantum theory (as opposed to a classical theory) one can use the language of {\em non-commutative}
operator algebras and their state space. Symmetries are then incorporated to this setting via a representation of the corresponding
group in terms of automorphisms of the operator algebra. These representations are typically implemented in terms of unitary representations of the
group
(see, e.g., \cite[Chapters~2 and 3]{bBratteli02} or \cite[Part~I]{bLandsman98}).
One of the conceptual advantages of (non-commutative) C*-algebras
is the neat distinction between
the abstract algebra, whose self-adjoint elements correspond to observables, and its state space and the corresponding
representations on a concrete Hilbert space. This point of
view particularly pays off in Quantum Field Theory,
where there is an abundance of inequivalent representations associated with
abstract observables
(cf. \cite{bEmch72}; see also Subsection~\ref{subsec:LQP} below).

\subsection{Type~$I$ algebras and Quantum Mechanics}
\label{subsec:QM}

The most elementary example of a type~$I$ factor is $\BH$, where $\cH$ is a finite or (separable) infinite dimensional Hilbert space.
Many situations in Quantum Mechanics can be described in terms of this example. Pure states correspond in 
this context to minimal projections and mixed states are described in terms of normalized and positive trace class operators.

We begin by making precise the fact that $\BH$ is, in fact, the prototype of this kind of factors. It is illustrative
to give a sketch of the proof since it shows how the minimality condition is used.

\begin{proposition}
Let $\cN\subset\BH$ be a factor of type~$I$. Then there exist separable Hilbert spaces $\cK_1$ and $\cK_2$ and a unitary 
$U\colon\cH \to \cK_1 \otimes \cK_2$ with $U\cN U^* = \mathcal{B}(\cK_1) \otimes  \1$.
\end{proposition}

\begin{proof}
Let $\{P_j\}_{j \in J} \subset \cP(\cN)$ be a maximal family of mutually orthogonal minimal projections. By maximality it follows
that $\cH \cong \bigoplus_{j \in J} P_j \cH$. 
Moreover, by minimality of projections, all $P_i, P_j$ must be equivalent for any pair $i, j \in J$. 
Therefore, there are partial isometries $V_{1j} \in \cN$ with $V_{1j}V_{1j}^* = P_1$ and $V_{1j}^*V_{1j} = P_j$, $j \in J$. 
This implies that $\cN$ is generated by the set $\{V_{1j} \mid j \in J\}$ since we have
\begin{equation}
\cN \ni N = \sum_{i,j \in J} P_i N P_j = \sum_{i,j \in J}\lambda_{ij}V_{1i}^*V_{1j}\;,
\end{equation}
where the coefficients $\lambda_{ij} \in \mathbb{C}$ are specified by the relation
\begin{equation*}
V_{1i}P_i N P_jV_{1j}^*  \in P_1 N P_1 = \mathbb{C}P_1\;,
\end{equation*}
which, again, uses the minimality of $P_1$. In fact, note that
\begin{equation*}
V_{1i}P_i N P_jV_{1j}^* = \lambda_{ij}P_1 \quad\mathrm{and hence}\quad P_i N P_j = \lambda_{ij}V_{1i}^*P_1V_{1j}^* = \lambda_{ij}V_{1i}^*V_{1j}\;.
\end{equation*}

Finally, consider the discrete set $\cJ = \{1,2,\ldots |J|\}$ with $|J|\in\N\cup\{\infty\}$ and define the unitary map
\[
U^*\colon \ell_2(\cJ,P_1\cH) \to \cH
\]
by means of $U^*\xi := \sum_j V_{1j}^* \xi_j$, where $\xi = (\xi_j)_{j=1}^{|J|} \in \ell_2(\cJ,P_1\cH)$. 
Using now the equivalence
$\ell_2(\cJ,P_1\cH) \cong \ell_2(\cJ) \otimes P_1\cH $
one can show that the algebra generated by 
$\{UV_{1j} U^* \mid j \in J \}$ is isomorphic to $\mathcal{B}(\ell_2(\cJ)) \otimes \1$, because
\begin{equation*}
UV_{1i}^*U^*U V_{1j}U^*= UV_{1i}^*V_{1j}U^* \cong E_{ij} \otimes \1
\end{equation*}
where $\{ E_{ij} \mid i,j \in J \}$ is a set of matrix units in $\ell_2(\cJ)$. 
\end{proof}

\begin{remark}
From the results mentioned in Section~\ref{sec:2} it is clear that $\BH$ with dim$\cH=\infty$ is an infinite as well as properly infinite 
C*-algebra. Nevertheless, observe that the structure of type~$I$ factors allows to have subalgebras of F\o lner type.
For instance, take two non-commuting range one projections $P, Q \in \mathcal{B}\left(\mathcal{H}\right)$, 
the von Neumann algebra generated by them will be finite-dimensional, and hence F\o lner. Note that this reasoning is not possible
in the context type~$III$ von Neumann algebras.
\end{remark}

\subsection{The CAR-algebra}\label{subsec:CAR}

In this section we give a proof that the C$^*$-algebras associated to the canonical anti-commutation relations (CAR-algebras) are, in fact, 
F\o lner C$^*$-algebras. We begin by recalling its definition and some standard properties (see, e.g., \cite[Section~5.2.2]{bBratteli02}).  

Let $\fh$ be a complex separable Hilbert space with scalar product
$\langle\cdot,\cdot\rangle$. We denote by CAR$(\fh)$ the algebraically unique C*-algebra generated by
$\1$ and $a(f)$, $f\in \fh$, such that the following relations hold:
\begin{itemize}
\item[{\rm (i)}] The map $\fh\ni f\mapsto a(f)$ is antilinear.
\item[{\rm (ii)}] $a(f_1) a(f_2)+a(f_2)a(f_1)=0\;$, $f_1,f_2\in \fh\,$.
\item[{\rm (iii)}] $a(f_1) a(f_2)^*+a(f_2)^*a(f_1)=\langle f_1,f_2\rangle \1\;$, $f_1,f_2\in \fh\,$.
\end{itemize}

The algebra CAR$(\fh)$ is simple, has a unique tracial state and satisfies $\|a(f)\|=\|f\|$ for any $f\in\fh$.
In the proof of the next theorem we exploit the finite approximation structure of the CAR-algebra.

\begin{proposition}
Let $\fh$ be a complex separable Hilbert space. Then CAR$(\fh)$ is a F\o lner C$^*$-algebra and its unique 
tracial state is amenable.
\end{proposition}
\begin{proof}
 If dim $\fh=n<\infty$, then CAR$(\fh)\cong M_{2^n}(\C)$ and hence F\o lner because it is finite dimensional.
 If dim $\fh=\infty$ we may describe the CAR-algebra as a uniformly hyper-finite algebra of type $2^\infty$ (see \cite[III.5.4]{bDavidson96}).
 In fact, CAR$(\fh)$ is the inductive limit of finite-dimensional algebras $\cA_n\cong M_{2^n}(\C)$ with injective 
 embedding
 \[
   \fA_n\ni A\mapsto \begin{pmatrix}
                            A & 0 \\ 0 & A
                           \end{pmatrix}
              \in  \fA_{n+1} \;.
 \]
 Consider the *-algebra $\fA:=\cup_{n=1}^\infty\fA_n$, which is dense in CAR$(\fh)$. We will prove that $\fA$ is algebraically
 amenable (cf. Definition~\ref{def:alg-amenable}) and therefore, by Theorem~\ref{teo:alg-amenable-Foelner}, we conclude
 that CAR$(\fh)$ is a F\o lner C$^*$-algebra. Define the finite dimensional subspaces (in fact subalgebras) $W_k:=\fA_k$, $k\in\N$. 
 Then, since any $A\in\fA$ is contained in $\fA_{k_0}$ for some $k_0\in\N$ we conclude that for any $k\geq k_0$ 
 we have $A W_k\subset W_k$, and therefore dim$(AW_k+W_k)=\mathrm{dim}(W_k)$ and
 \[
  \lim_{k\to\infty}\frac{\mathrm{dim}(AW_k+W_k)}{\mathrm{dim}(W_k)}=1\;.
 \]
 Finally, since CAR$(\fh)$ has a unique tracial state it must be amenable. 
\end{proof}

\subsection{Local Quantum Physics}
\label{subsec:LQP}

In this subsection we address several manifestations of infinity that appear in quantum field theory.
For this analysis we use the axiomatic approach proposed by
Haag and Kastler in the sixties using the language of operator algebras
(see, e.g., \cite{Haag64,bHaag92,RobertsIn04,bBaumgaertel95}), usually known as Algebraic Quantum Field Theory or Local Quantum Physics.
In this formulation the observables become the primary objects of the
theory and are described by selfadjoint elements in an abstract C$^*$-algebra.
 Here one considers the observables to be
localized in spacetime, which, in this article, we restrict to be the 4-dimensional Minkowski space. 
The fundamental object of study  is a {\em net of von Neumann algebras}
labeled by spacetime regions in $\R^4$. Concretely, we consider the {\em index set}
\[ 
 \cI:=\{\cO\subset\R^4 \mid \cO \;\;\mathrm{open~and~bounded~region~in~Minkowski~space}\}
\]
and a {\em net  of von Neumann algebras} is denoted by 
\[
 \cI\ni\cO\mapsto \cN(\cO)\subset \cB(\cH)\;.
\]
Associated with this net we can define the {\em global algebra} by
$
\cR: =\Big(\mathop{\cup}\limits_{\cO\in\cI}\cN(\cO)\Big)'' \;.
$

We begin by recalling the axioms of the {\em vacuum representation}.
The axioms specifying this representation of the net $\cI\ni\cO\mapsto \cN(\cO)$ are
physically motivated and have physical and mathematical consequences. These rules formalize 
general principles of relativistic quantum mechanics like, e.g., Poincar\'e covariance or causality. 
Characteristic for the vacuum state is its invariance under the Poincar\'e group and 
the (relativistic) spectrum condition. 

\begin{enumerate}
 \item[(A1)] {\em Isotony:} If $\cO_1\subset\cO_2$ then $\cN(\cO_1)\subset\cN(\cO_2)$.
 \item[(A2)] {\em Additivity:} If $\cO=\cup_j\cO_i$ then $\cN(\cO)=\Big(\cup_j\cN(\cO_j)\Big)''$.
 \begin{itemize}
  \item[(A2$'$)] {\em Weak additivity:} For each $\cO_0\in\cK$ we have 
                 $\Big(\mathop{\cup}\limits_{a\in\R^4}\cN(a+\cO_0)\Big)''=\cR$.
 \end{itemize}

 \item[(A3)] {\em Causality:} If $\cO_1\perp\cO_2$ (i.e., $\cO_1$ and $\cO_2$ are causally disjoint), then $\cN(\cO_1)\subset\cN( \cO_2)'$.
 \item[(A4)] {\em Covariance:} There is a strongly continuous unitary representation of the universal cover of the 
 proper orthocronous Poincar\'e group $\cG:=\R^4\rtimes\mathrm{SL}(2,\C)$, $U\colon\cG\to\cU(\cH)$ such that
 \[
  \cN(g\cO)=\alpha_g(\cN(\cO))=U(g)\cN(\cO) U(g)^{-1}\;,\quad \alpha_g\in\mathrm{Aut}\,{\cR}\;,\; g\in\cG\;.
 \]

\item[(A5)] {\em Spectrum condition:} The spectrum of the generators of the space-time translations is contained 
in the closed forward light cone, i.e., 
\[ 
 \sigma\Big(U(\R^4)\Big)\subset V_+ \;.
\]
\

 \item[(A6)] {\em Existence of a vacuum vector:} There exists a unit vector $\Omega\in\cH$ (called the vacuum vector) such that
 \[
  \Big(\cup_{\cO\in\cK}\cN(\cO)\Big)\Omega \;\;\mathrm{is~dense~in~}\cH\quad\mathrm{and}\quad U(g)\Omega=\Omega\;,g\in\cG\;.
 \]
\end{enumerate}

For concrete examples of nets satisfying these axioms we refer to the free-net construction in \cite{BJLL95,Lledo04}
as well as references therein.
An immediate and surprising consequence of this set of axioms is the so-called {\em Reeh-Schlieder Theorem}. 

\begin{theorem}
Let $\cI\ni\cO\mapsto \cN(\cO)\subset\BH$ be a net satisfying the axioms of the vacuum representation. For every nonempty region
$\cO\in\cI$ the vacuum vector $\Omega$ is cyclic and separating for $\cN(\cO)$, i.e., the set $\cN(\cO)\Omega\subset\cH$
is dense in $\cH$ and, for any local operator $N\in\cN(\cO)$, one has that $N\Omega=0$ implies $N=0$. 
\end{theorem}

This result implies, in particular, that any nonzero local projection in $\cN(\cO)$ has (for any nonempty $\cO\in\cI$) 
a nonzero expectation value in the vacuum. Moreover, this result also shows that the vacuum in quantum field theory is entangled for any pair of local
algebras $\cN(\cO_1)$, $\cN(\cO_2)$ with $\cO_1\perp\cO_2$. 
We refer, e.g., to \cite[\S 1.3]{bBaumgaertel95} for a complete proof of the Reeh-Schlieder theorem
which makes explicit use of the covariance axiom, weak additivity
and the spectrum condition. For additional motivation, results and references see
\cite{Redhead95,SummersIn11}.

The next result is known as {\em Wightman's inequality}.
Let $\cO_1,\cO_2\in\cI$ and denote by $\cO_1\Subset \cO_2$ if $\cO_1\subset \cO_2$ and the distance of $\cO_1$ to the boundary 
of $\cO_2$ is positive, i.e., $\mathrm{dist}(\cO_1,\partial\cO_2)>0$.
\begin{theorem}
Let  $\cI\ni\cO\mapsto \cN(\cO)\subset\BH$ be a net satisfying the axioms of the vacuum representation and such that the global algebra
$\cR$ is non-Abelian. Then for any $\cO_1\Subset \cO_2$ we have that $\cN(\cO_1)\subsetneqq\cN(\cO_2)$.
\end{theorem}

This result implies that for each $\cO\in\cI$ the local algebras $\cN(\cO)$ are necessarily infinite dimensional, since 
for $\cO_1\Subset \cO_2$ we must have $\mathrm{dim}_\C\cN(\cO_1)<\mathrm{dim}_\C\cN(\cO_2)$. A complete proof of Wightman's inequality
can be found in \cite[\S 1.4]{bBaumgaertel95} which uses explicitly the isotony axiom as well as covariance and weak additivity.

Local algebras are not only infinite dimensional, they are typically type~$III$ (showing the highest degree 
of infiniteness). The change in relativistic quantum mechanics to a net of algebras localized in spacetime regions $\cO\in\cI$ forces
the radical change to type~$III$ (as opposed to a type~$I$ description in quantum mechanics). For specific regions such as a 
space-like wedge or for theories which, in addition, have conformal covariance, it can be even shown that
the local algebras correspond to the unique hyperfinite type~$III_1$ factor (see, e.g., \cite[Section~V.6]{bHaag92} for details).

We conclude this section mentioning Borchers property which implies that, generically, local algebras are almost type~$III$. 
This property, which is strongly based on the positivity of the energy, is enough in many applications.
For a proof we refer to \cite[\S 1.11 and 1.12]{bBaumgaertel95}. Before stating the next result, recall from Section~\ref{sec:2} that
for a von Neumann algebra is of type~$III$ all nonzero projections are equivalent to $\1$.

\begin{theorem}
Let $\cI\ni\cO\mapsto \cN(\cO)\subset\BH$ be a net satisfying the axioms of the vacuum representation, with unique vacuum vector $\Omega$.
Assume $\cO_1,\cO_2\in\cI$ satisfy $\cO_1\Subset \cO_2$ and that there exists an $\cO\in\cJ$ with $\cO\subset\cO_1^\perp\cap\cO_2$. Then
for any nonzero projection $P\in\cN(\cO_1)$ we have
\[
 P\sim\1\;\mathrm{mod}\;\cN(\cO_2)\;.
\]
\end{theorem}

As an application of type~$III$ structure appearing in quantum field theory we refer, e.g., to the explanation of Fermi's two atom system
(cf. \cite{BY04,Yngvason05}).

\subsection{The theory of superselection sectors}
\label{subsec:SST}

The theory of superselection sectors allows from an analysis of a physically motivated family of states
to understand three central aspects in elementary particle physics: the composition of charges, the classification 
of particle statistics and the charge conjugation. In this final subsection we will mention briefly 
the role that Cuntz algebras play in this frame, confirming again the importance of properly infinite C*-algebras
in quantum field theory. The theory of superselection sectors
as stated by the Doplicher-Haag-Roberts~(DHR) selection criterion
\cite{bHaag92,DHR69a,DHR69b}, is formulated in the frame of local quantum physics and
led to a profound body of work, culminating in
the general Doplicher-Roberts~(DR) duality theory for compact groups
\cite{Doplicher89b}.

The DHR criterion selects a distinguished class of ``admissible''
representations of a quasilocal algebra $\cA$ of observables,
which has trivial center
$\cZ:=\cZ(\cA)=\C\1$.
This class of representations specifies a so-called DR-category $\cT$,
which is a full subcategory of the category of
endomorphisms of the C*-algebra $\cA$. 
Furthermore, from this endomorphism category $\cT$ the
DR-analysis constructs a C*-algebra $\cF\supset\cA$ together with a compact group action 
$\alpha\colon\cG\ni g\to\alpha_{g}\in\mathrm{Aut}(\cF)$
such that:
\begin{itemize}
\item
$\cA$ is the fixed point algebra of this action;
\item
$\cT$ coincides with the category of all ``canonical
endomorphisms" of $\cA$, associated with the pair
$\{\cF,\alpha_{\cG}\}$.
\end{itemize}

Physically, $\cF$ is identified as a field algebra 
and $\cG$ with a global gauge group of the system.
The pair $\{\cF,\alpha_{\cG}\}$, which we call 
{\em Hilbert C*-system} (see below for a precise definition),
is uniquely determined by $\cT$ up to isomorphisms. Conversely, $\{\cF,\alpha_{\cG}\}$
determines uniquely its category of all canonical endomorphisms.
Therefore $\{\cT,\cA\}$ can be seen as the abstract side
of the representation category of a compact group, while 
$\{\cF,\alpha_{\cG}\}$ corresponds
to the concrete side of the representation category of 
$\cG$, and, roughly, any irreducible representations
of $\cG$ is explicitly realized within 
the Hilbert C*-system.
One can state the equivalence of the ``selection
principle", given by
$\cT$
and the ``symmetry principle", given by the compact group
$\cG$. This is one of the crucial theorems of the
Doplicher-Roberts theory (see also \cite{bBaumgaertel95,FredenhagenIn92,RobertsIn04}
for additional results and motivation). 

We conclude explaining the structure of Hilbert C*-systems.
These are, roughly speaking, a very special  
type of C*-dynamical system $(\mathcal{F}, \alpha_{\cG})$ that,
in addition, contain the information of the representation category
of the compact group $\mathcal{G}$. We denote the dual object of $\cG$ 
by $\widehat{\cG}$, which is defined as the set of (unitary) equivalence
classes of continuous irreducible unitary representations of $\cG$ (on complex separable Hilbert spaces).
A Hilbert space $\cH\subset\cF$, 
where $\cF$ is a unital C*-algebra, is called {\it algebraic} if the scalar product 
$\langle\cdot,\cdot\rangle$ of $\cH$
is given by $\langle A,B\rangle\1 := A^*B$ for $A,\; B\in\cH$. Henceforth,
we consider only finite-dimensional algebraic Hilbert spaces. 
The support of $\cH$ is defined by
$\hbox{supp}\,\cH:=\sum_{j=1}^{d}\Psi_j\Psi_{j}^*$,
where $\{\Psi_j\,\big|\,
j=1,\ldots,\,d\}$ is any orthonormal basis of $\cH$. 
We consider here only algebraic Hilbert space $\cH$
with ${\rm supp}\,\cH=\1$. For any $D\in\widehat{\cG}$ consider the 
following projection on $\cF$
\[
 \Pi_D(\cdot):=\int_\cG \chi_D(g)\alpha_g(\cdot)\,dg\;,
\]
where $\chi_D$ is the modified character of the class $D$, i.e., 
$\chi_D(g):=\dim (D) \; \mathrm{Tr}(D(g))$. The subspaces $\Pi_D$, $D\in\widehat{\cG}$,
are called spectral subspaces of $\cF$. Note that if one chooses the trivial representation
$\iota\in\widehat{\cG}$, then the corresponding spectral subspace is the fixed point algebra
\[
 \Pi_\iota(\cF):=\{A\in\cF\mid \alpha_g(A)=A\;,\;g\in\cG\}\;,
\]
which in our context turns out to coincide with the C*-algebra $\cA$.

\begin{definition}\label{def:HCS}
A C*-dynamical $\{\cF,\alpha_{\cG}\}$ with a compact group $\cG$ is called a {\bf Hilbert C*-system} if
for each $D\in\widehat{\cG}$
there is an algebraic Hilbert space 
$\cH_{D}\subset\Pi_{D}\cF,$
such that $\alpha_{\cG}$
acts invariantly on $\cH_{D},$  
and the unitary representation
$\alpha_\cG|_{\cH_{D}}$ 
is in the equivalence class 
$D\in\widehat{\cG}.$
\end{definition}

Note that any algebraic Hilbert space $\cH_D$, $D\in\widehat{\cG}$, generates 
a Cuntz algebra $\cO_n$ with $n=\dim D$ which are all subalgebras of the field algebra
$\cF$. Moreover, any algebraic Hilbert space $\cH_D$ specifies a canonical endomorphism 
of the fixed point algebra by
\[
 \rho_D(A)=\sum_{i=1}^n \Psi_i A \Psi_i^*\;,
\]
where  $\{\Psi_i\,\big|\,
i=1,\ldots,\,n\}$ is any orthonormal basis of $\cH_D$. Since the supp$\cH_D=\1$ the canonical 
endomorphisms are also unital, i.e., $\rho_D(\1)=\1$.

\begin{remark}
In the DR-theory the center $\cZ$ of the C*-algebra $\cA$ 
plays a special role. If $\cA$ corresponds
to the inductive limit of a net of local C*-algebras indexed by
open and bounded regions of Minkowski space, then the triviality
of the center of $\cA$ is a consequence of standard assumptions
on the net of local C*-algebras. But, in general, the C*-algebra appearing 
in the DR-theorem does not need to be a quasilocal algebra and, in fact,
one has to assume explicitly that $\cZ=\C\1$
in this context (see \cite[Theorem~6.1]{Doplicher89b}).
Therefore from a systematical point of view it is natural to study the
properties and structural modifications of this
rich theory if one assumes
the presence of a nontrivial center $\cZ\supset\C\1$.
From a physical point of view one can interpret the elements of the center $\cZ$ of $\cA$ as
classical observables contained in the quasilocal algebra. Nevertheless the effect of the
presence of classical observables in superselection theory requires a more careful analysis
of the corresponding fundamental axioms.
We refer to \cite{Lledo01a,Lledo97b,BL04}
for an analysis of the DR-duality theory in the case the 
relative commutant of the corresponding
Hilbert C*-system satisfies the following minimality condition:
\[
 \cA'\cap\cF=\cZ\;.
\]
Concrete realization of these systems in terms of Cuntz-Pimsner algebras, a class of properly infinite
C*-algebras generalizing Cuntz algebras, can be found in \cite{LV09}.
\end{remark}


\begin{thebibliography}{99.}


\bibitem{ALLW-1}
P.~Ara, K.~Li, F.~Lled\'o and J.~Wu, \textit{Amenability of coarse spaces and $\mathbb{K}$-algebras}, 
Bull. Math. Sci. {\bf 8} (2018) 257-306.

\bibitem{ALLW-2}
P.~Ara, K.~Li, F.~Lled\'o and J.~Wu, \textit{Amenability and uniform Roe algebras}, J. Math. Anal. Appl. {\bf 459} (2018) 686-716. 

\bibitem{AL14}
P.~Ara and F.~Lled\'o, \textit{Amenable traces and F\o lner $C^*$-algebras}, 
Expo. Math. {\bf 32} (2014) 161--177.


\bibitem{Arveson94}
W.~Arveson,
{\em $C^*$-algebras and numerical linear algebra},
J. Funct. Anal. \textbf{122} (1994) 333--360.

\bibitem{bBaumgaertel95}
H.~Baumg\"artel, {\em Operatoralgebraic Methods in Quantum Field Theory. A Series of
  Lectures}, Akademie Verlag, Berlin, 1995.
  
  
\bibitem{BJLL95}
H.~Baumg\"artel, M.~Jurke and F.~Lled\'o, {\em On free nets over Minkowski space}, Rep. Math. Phys {\bf 35} (1995) 101-127.

\bibitem{Lledo97b}
H.~Baumg\"artel and F.~Lled\'o, {\em Superselection structures for C*-algebras
  with nontrivial center}, Rev. Math. Phys. \textbf{9} (1997) 785-819.

\bibitem{Lledo01a}
H.~Baumg\"artel and F.~Lled\'o, 
  {\em An application of DR-duality theory for compact groups to
  endomorphism categories of C*-algebras with nontrivial center}, Fields Inst.
  Commun. \textbf{30} (2001) 1-10.

\bibitem{BL04}
H.~Baumg\"artel and F.~Lled\'o, {\em Duality of compact groups and Hilbert
  C*-systems}, Int. J. Math. {\bf 15} (2004) 759--812.

\bibitem{Bedos97}
E.~B\'{e}dos, \textit{On F{\o}lner nets, Szeg\"o's theorem
and other eigenvalue distribution theorems},
Expo. Math. \textbf{15} (1997) 193--228.
Erratum: Expo. Math. \textbf{15} (1997) 384.

\bibitem{Bedos95}
E.~B\'edos, \textit{Notes on hypertraces and $C^*$-algebras},
J. Operator Theory \textbf{34} (1995) 285--306.

\bibitem{bBratteli02}
O.~Bratteli and D.W. Robinson, {\em Operator Algebras and Quantum Statistical
  Mechanics. Vols.~1 and 2}, Springer Verlag, Berlin, 2002.

\bibitem{bBrown08}
N.P.~Brown and N.~Ozawa, \textit{$C^*$-Algebras and Finite-Dimensional
Approximations}, American Mathematical Society, Providence,
Rhode Island, 2008.

\bibitem{BY04}
D.~Buchholz and J.~Yngvason, {\em 
There are no causality problems for Fermi's two-atom system},
Phys. Rev. Lett. {\bf 73} (1994) 613--616.

\bibitem{Cec01} T. Ceccherini-Silberstein, {\em Around amenability}, Pontryagin Conference, 8, Algebra (Moscow, 1998). 
J. Math. Sci. (New York) {\bf 106} (2001) 3145--3163.

\bibitem{Silberstein-Grigorchuk-Harpe-99}
T. Ceccherini-Silberstein, R. Grigorchuk and P. de la Harpe, {\em Amenability and paradoxical decomposition
for pseudogroups and for discrete metric spaces}, Proc. Steklov Inst. Math. {\bf 224} (1999) 57--97.\\
(Updated version also under {\tt http://arxiv.org/abs/1603.04212})

\bibitem{bDavidson96}
K.R. Davidson, {\em {\rm C}$^*$-Algebras by Example}, American Mathematical
  Society , Providence, Rhode Island, 1996.

\bibitem{Day57} M. Day, {\em Amenable semigroups}, Illinois J. Math. {\bf 1} (1957) 509--544.

\bibitem{DHR69a}
S.~Doplicher, R.~Haag, and J.E. Roberts, {\em Fields, observables and gauge
  transformations I}, Commun. Math. Phys. \textbf{13} (1969) 1--23.

\bibitem{DHR69b}
S.~Doplicher, R.~Haag, and J.E. Roberts, 
  {\em Fields, observables and gauge transformations II}, Commun. Math.
  Phys. \textbf{15} (1969) 173--200.

\bibitem{Doplicher89a}
S.~Doplicher and J.E. Roberts, 
  {\em Endomorphisms of {\rm C}*-algebras, cross products and duality
  for compact groups}, Ann. Math. \textbf{130} (1989) 75--119.

\bibitem{Doplicher89b}
S.~Doplicher and J.E. Roberts, 
  {\em A new duality for compact groups}, Invent. Math. \textbf{98}
  (1989) 157--218.

\bibitem{Elek03}
G.~Elek, \textit{The amenability of affine algebras},
J. Algebra {\bf 264} (2003) 469-478.


\bibitem{bEmch72}
G.G.~Emch, {\em Algebraic Methods in Statistical Mechanics and Quantum Field
  Theory}, Wiley Interscience, New York, 1972.


\bibitem{FredenhagenIn92}
K.~Fredenhagen, {\em Observables, superselection sectors and gauge groups}, In
  {\em New Symmetry Principles in Quantum Field Theory}, J.~Fr\"ohlich et~al.
  (ed.), Plenum Press, New York, 1992.

\bibitem{Gromov99}
M.~Gromov, {\em Topological Invariants of Dynamical Systems
and Spaces of Holomorphic Maps: I}, 
Math. Phys. Anal. Geom. {\bf 2} (1999) 323-415.

\bibitem{bHaag92}
R.~Haag, {\em Local Quantum Physics}, Springer Verlag, Berlin, 1992.

\bibitem{Haag64}
R.~Haag and D.~Kastler, {\em An algebraic approach to quantum field theory}, 
J. Math. Phys. \textbf{5} (1964) 848--861.

\bibitem{KL15} D.~Kerr and H.~Li, \textit{Ergodic Theory. Independence and Dichotomies},  
Springer Monographs in Mathematics. Springer, Cham, 2016.

\bibitem{Keyl06}
M. Keyl, T. Matsui, D. Schlingemann and R.F. Werner,
{\em Entanglement, Haag-duality and type properties of infinite quantum spin chains},
Rev. Math. Phys. {\bf 18} (2006) 935--970.

\bibitem{bLandsman98}
N.P. Landsman, {\em Mathematical Topics between Classical and Quantum Mechanics},
  Springer, New York, 1998.

\bibitem{Lledo04}
F.~Lled\'o, {\em Massless relativistic wave equations and quantum field theory}, Ann. H. Poincaré {\bf 5} (2004) 607-670. 
  
\bibitem{Lledo13}
F.~Lled\'o,
\textit{On spectral approximation, F\o lner sequences and crossed products},
J. Approx. Theory {\bf 170} (2013) 155--171.

\bibitem{LV09}
F.~Lled\'o and E. Vasselli, {\em Realization of Hilbert C*-systems in terms of Cuntz-Pimsner algebras}, Int. J. Math. {\bf 20}
(2009) 751--790.

\bibitem{LledoYakubovich13}
F.~Lled\'o and D.~Yakubovich,
\textit{F\o lner sequences and finite operators}, 
J. Math. Anal. Appl. {\bf 403} (2013) 464--476.

\bibitem{Naaijkens11}
P.~Naaijkens,
\textit{Localized endomorphisms in Kitaev’s toric code on the plane},
Rev. Math. Phys., {\bf 23} (2011) 347--373, 2011

\bibitem{Neumann29}
J.~von~Neumann, \textit{Zur allgemeinen Theorie des Masses}, Fund. Math. \textbf{13} (1929) 73--116.

\bibitem{bPaterson88} A.L.~Paterson, \textit{Amenability}, American
Mathematical Society, Providence, Rhode Island, 1988.

\bibitem{Redhead95}
M.~Redhead, \textit{More ado about nothing}, Fund. Phys. \textbf{25} (1995) 123-137.

\bibitem{RobertsIn04}
J.E.~Roberts, {\em More lectures on algebraic quantum field theory}, In {\em 
  Noncommutative Geometry}, S.~Doplicher and R.~Longo (eds.), 
  Lecture Notes in Mathematics Vol.~1831, Springer Verlag, Berlin, 2004.

\bibitem{bRunde02} V.~Runde, \textit{Lectures on Amenability}, Springer, Berlin, 2002.

\bibitem{SummersIn11} S.~Summers, {\em Yet more ado about nothing}, In {\em Deep Beauty. Understanding the Quantum World Through Mathematical
Innovation},
H.~Halvorson (Ed.), Cambridge University Press, Cambridge, 2011; pp. 317--341.

\bibitem{Yngvason05}
J.~Yngvason, {\em The role of type~$III$ factors in quantum field theory},
Rep. Math. Phys. {\bf 55} (2005) 135--147. 

\bibitem{Wagon16} 
G.~Tomkowicz and S.~Wagon, \textit{The Banach-Tarski Paradox (second edition)}, Cambridge University Press, Cambridge, 2016.



\end{thebibliography}
\end{document}